\documentclass[a4paper, 10pt]{article}

\usepackage{amsthm}
\usepackage{amsmath}
\usepackage{amsfonts}
\usepackage{amssymb}
\usepackage{bbm}

\usepackage{cleveref}
\usepackage{amscd}
\usepackage{tikz-cd}

\newcommand{\myset}[1]{\{1,\dots,#1\}}
\newcommand{\id}{\mathbbm{1}}
\newcommand{\card}[1]{|{#1}|}
\newcommand{\V}{\mathcal{V}}
\newcommand{\U}{\mathcal{U}}
\newcommand{\comment}[1]{}

\newtheorem{proposition}{Proposition}
\newtheorem{theorem}{Theorem}
\newtheorem{lemma}{Lemma}
\theoremstyle{definition}

\newtheorem{example}{Example}
\newtheorem{remark}{Remark}

\DeclareMathOperator{\Hom}{Hom}
\DeclareMathOperator{\Ker}{Ker}
\DeclareMathOperator{\Img}{Im}
\DeclareMathOperator{\Aut}{Aut}
\DeclareMathOperator{\wt}{wt}
\DeclareMathOperator{\M}{M}
\DeclareMathOperator{\soc}{soc}

\newcommand{\Z}{\mathbb{Z}}

\usepackage{titling}
\usepackage[affil-it]{authblk}
\title{Geometric approach to the MacWilliams Extension Theorem for codes over modules}
\author{Serhii Dyshko
	\thanks{Electronic address: \texttt{dyshko@univ-tln.fr}}}
\affil{Institut de math\'ematiques de Toulon, Universit\'e de Toulon, France}
\date{}

\begin{document}
\maketitle

\begin{abstract}
	The MacWilliams Extension Theorem states that each linear Hamming isometry of a linear code extends to a monomial map.
	In this paper an analogue of the extension theorem for linear codes over a module alphabet is observed.
	A geometric approach to the extendability of isometries is described. For a matrix module alphabet we found the minimum length of a code for which an unextendable Hamming isometry exists. We also proved an extension theorem for MDS codes over a module alphabet.
\end{abstract}

\section{Introduction}
The famous MacWilliams Extension Theorem states that each linear isometry of a linear code extends to a monomial map. The result was originally proved in the MacWilliams' Ph.D. thesis, see \cite{macwilliams-phd61}. Later, the result was generalized for codes over modules. A summary is given below.

Let $R$ be a ring with identity and let $A$ be a finite left $R$-module. Consider a module $A^n$ with the Hamming metrics and a code $C \subseteq A^n$ that is a left $R$-submodule. For two $R$-modules $A$ and $B$ let $\Hom_R(A,B)$ denote the set of $R$-module homomorphism from $A$ to $B$.
Call a map $f \in \Hom_R(A^n,A^n)$ \emph{monomial}, if there exist a permutation $\pi \in S_n$ and automorphisms $g_1,\dots, g_n \in \Aut_R({A})$, such that, for any $a \in {A}^n$,
\begin{equation*}
f\big((a_1,\dots,a_n)\big) = (g_1(a_{\pi(1)}), \dots, g_n(a_{\pi(n)}))\;.
\end{equation*}

Note that a monomial map preserves the Hamming distance. It can be easily shown, that each isometry $f \in \Hom_R(A^n,A^n)$ is monomial.
We say that an alphabet $A$ has the \emph{extension property} if for any positive integer $n$, for any code $C \subseteq A^n$ each Hamming isometry $f \in \Hom_R(C, A^n)$ extends to a monomial map.

\begin{figure}[!ht]
	\centering
	\begin{tikzcd}
		A^n \arrow{rd}{h}&\\
		C \arrow[hook]{u}{\iota} \arrow[swap]{r}{f}& A^n 	
	\end{tikzcd}
	\caption{Extension property.}
	\label{fig-extendibility}
\end{figure}

Classical linear codes correspond to the case when $R = \mathbb{F}_q$ and $A = R$, where $\mathbb{F}_q$ is a finite field. As we mentioned before, the MacWilliams Extension Theorem states that the alphabet $A = R = \mathbb{F}_q$ has the extension property.

Apparently, not every module alphabet satisfies the extension property.
Recall the definition of a pseudo-injective module. A left $R$-module $A$ is called \emph{pseudo-injective}, if for each left $R$-submodule $B \subseteq A$ and for each two embeddings $\phi,\psi \in \Hom_R({B},{A})$ there exists an automorphism $h \in \Aut_R({A})$ such that $\psi = h\phi$.
\begin{figure}[!ht]
	\centering
	\begin{tikzcd}
		B \arrow[hook]{r}{\psi} \arrow[hook, swap]{dr}{\phi} & A\\
		& A \arrow[swap]{u}{h} 	
	\end{tikzcd}
	\caption{Pseudo-injectivity of a module $A$.}
	\label{fig-x-prop}
\end{figure}
In other words, $A$ is pseudo-injective if and only if any $R$-linear isomorphism between its submodules extends to an $R$-linear automorphism of $A$. Not all the $R$-modules are pseudo-injective.
\begin{example}
	Consider a $\Z$-module $A = \Z_2 \oplus \Z_4$. Let $M = \langle (0,2)\rangle$ and $N = \langle (1,0)\rangle$. Obviously, $M \cong N$, with isomorphism $\psi: (0,2) \rightarrow (1,0)$, but there is no isomorphism $\phi: A \rightarrow A$ such that $\phi = \psi$ on $M$. So, $A$ is not pseudo-injective.
\end{example}

Recall the \emph{socle} of $A$ is a submodule $\soc(A) \subseteq A$ that is equal to the sum of all simple submodules of $A$. A module is called \emph{simple} (or irreducible) if it does not contain any other submodules except zero and itself.
In \cite{wood-foundations} the author proved a general extension theorem for a pseudo-injective module alphabet with a cyclic socle and showed that these conditions are maximal.
\begin{theorem}\label{thm-wood-general-modules}
	If $A$ is an alphabet that is pseudo-injective and $\soc(A)$ cyclic, then $A$ has the extension property.
\end{theorem}

In the case, when $A = R$, in \cite{dinh-lopez-1, greferath, wood-foundations} the authors proved the extension theorem for Frobenius rings and showed the maximality of the condition.
\begin{theorem}\label{thm-frobenius-wood}
	Let $R$ be a Frobenius ring. Then the alphabet $A = R$ has the extension property.
\end{theorem}

In \cite{dinh-lopez, dinh-lopez-1} the extension problem for arbitrary ring and alphabet was partially translated to the case of matrix rings and matrix modules.
There the authors proved the existence of a general counterexample for codes over a matrix module alphabet. An explicit construction appeared in \cite{wood-foundations}.
\begin{theorem}[see \cite{wood-foundations}]\label{thm-wood-matrix-module}
	Let $R = \M_{m}(\mathbb{F}_q)$ be a ring of all $m \times m$ matrices over a finite field $\mathbb{F}_q$ and let $A = \M_{m \times k}(\mathbb{F}_q)$ be a left $R$-module of all $m \times k$ matrices over $\mathbb{F}_q$.
	
	If $k \leq m$, then the alphabet $A$ has the extension property.
	
If $k > m$, there exist a linear code $C \subset A^K$, $K = \prod_{i=1}^{k-1} (1 + q^i)$, and a map $f \in \Hom_R(C,A^K)$ that is a Hamming isometry, but there is no monomial map extending $f$.
\end{theorem}
Note that the theorem does not say if the given  counterexample has minimum possible code length.
In this paper we improve \Cref{thm-wood-matrix-module}. More precisely, in the context of matrix modules, we found the minimum code length for which an example of a code with unextendable isometry exists. It appear that such a code with minimal code length is similar to the code from \Cref{thm-wood-matrix-module}, described in \cite{wood-foundations}. In our previous work \cite{d1} we found the precise bound for the case $m = 1$, which corresponds to linear codes over a vector space alphabet.

Our main idea is to use a geometric approach. In \Cref{thm-isometry-criterium} we describe an unextendable isometries in terms of nontrivial solution of the isometry equation (\ref{eq-main-space-equation}), which is an equation of indicator functions of modules. In \cite{d1} we observed basic properties of this equation for the case of vector spaces and here, in \Cref{thm-matrix-module-bound} we describe some properties of the equation for matrix modules.

In \Cref{thm-mds-extension-theorem} we prove that the extension property holds for MDS codes over a module alphabet, when the dimension of a code does not equal 2. Despite the general result of \Cref{thm-wood-general-modules}, for MDS codes the extension theorem holds for arbitrary finite $R$-module alphabet.

\section{Extension criterium}
Let $W$ be a left $R$-module isomorphic to $C$.
Let $\lambda \in \Hom_R(W, {A^n})$ be a map such that $\lambda(W) = C$.
Present the map $\lambda$ in the form $\lambda = (\lambda_1,\dots, \lambda_n)$, where $\lambda_i \in \Hom_R(W,{A})$ is a projection on the $i$th coordinate, for $i \in \myset{n}$. Consider the following modules, for $i \in \myset{n}$,
\begin{equation*}
	V_i = \Ker \lambda_i \subseteq W\;.
\end{equation*}

Let $f: C \rightarrow A^n$ be a homomorphism of left $R$-modules.
Define $\mu = f\lambda \in \Hom_R(W,A^n)$ and denote
\begin{equation*}
U_i = \Ker \mu_i \subseteq W\;.
\end{equation*}

\begin{figure}[!ht]
	\centering
	\begin{tikzcd}
		W \arrow{r}{\lambda} \arrow[swap]{rd}{\mu} & C \arrow{d}{f}\\
		& A^n  	
	\end{tikzcd}
	\caption{The maps $\lambda$ and $\mu$.}
	\label{fig-parametriztion}
\end{figure}

Denote the tuples of modules $\V = (V_1, \dots, V_n)$ and $\U = (U_1, \dots, U_n)$.
We say that $\V = \U$ if they represent the same multiset of modules. In other words, $\V = \U$ if and only if there exists $\pi \in S_n$ such that for each $i \in \myset{n}$, $U_i = V_{\pi(i)}$. Recall the indicator function of a subset $Y$ of a set $X$ is a map $\id_Y: X \rightarrow \{0,1\}$, such that $\id_Y(x) = 1$ if $x \in Y$ and $\id_Y(x) = 0$ otherwise.

\begin{proposition}\label{thm-isometry-criterium}
	The map $f \in \Hom_R (C,A^n)$ is a Hamming isometry if and only if the following equality holds,
	\begin{equation}\label{eq-main-space-equation}
	\sum_{i=1}^n \id_{V_i} = \sum_{i=1}^n \id_{U_i}\;.
	\end{equation}	
	If $f$ extends to a monomial map, then $\V = \U$.
	If $A$ is pseudo-injective and $\V = \U$, then $f$ extends to a monomial map. 
\end{proposition}
\begin{proof}
	Prove the first part. By definition, the map $f$ is a Hamming isometry if for each $a \in C$, $\wt(f(a)) = \wt(a)$, or, equivalently, $f$ is an isometry if and only if for each $w \in W$, $\wt(\lambda(w)) = \wt(\mu(w))$. Note that
	for any $w \in W$, $n - \wt(\lambda(w)) = \sum_{i=1}^{n} (1 - \wt(\lambda_i(w))) = \sum_{i=1}^n \id_{\Ker \lambda_i} (w)$ and the same for the map $\mu$. Hence \cref{eq-main-space-equation} holds.
	
	Prove the second part. Consider any two maps $\sigma, \tau \in \Hom_R(W,A)$. 
	If there exists $g \in \Aut_R({A})$ such that $\sigma = g\tau$ then $\Ker \sigma = \Ker \tau$.
	
	Let $\Ker \sigma = \Ker \tau = N \subseteq W$ and let $A$ be pseudo-injective. The corresponding homomorphisms defined on the quotient module $\bar{\sigma}, \bar{\tau}: W/N \rightarrow A$ are injective. Using the property of $A$, there exists $h\in \Aut_R(A)$ such that $\bar{\sigma} = h \bar{\tau}$. It is easy to check that $\sigma = h \tau$.
\end{proof}

Let a pair of tuples of modules $(\U,\V)$ be a solution of \cref{eq-main-space-equation}. If $\U= \V$ then we call the solution \emph{trivial}. \Cref{thm-isometry-criterium} gives a relation between trivial solutions of \cref{eq-main-space-equation} and extendable isometries.

\begin{remark}As it was noted in \cite{dinh-lopez-1} and \cite{wood-foundations}, the property of pseudo-injectivity is necessary in the statement. Assuming that the alphabet is not pseudo-injective means that the extension property fails even if the length of a code is 1.
\end{remark}

\section{General extension theorem}
In this section we show how to use the approach from the previous section to prove \Cref{thm-wood-general-modules}.
Recall that a left $R$-module $M$ is called \emph{cyclic} if there exists a generator element $x \in M$, such that $M = Rx = \{rx \mid r \in R \}$.
\begin{lemma}\label{lemma-cyclic-if-covering}
	An $R$-module $M$ is not cyclic if and only if there exist submodules $\{0\} \subset E_1, \dots, E_r \subset M$ such that $M = \bigcup_{i=1}^r E_i$.
\end{lemma}
\begin{proof}
	Assume that there exists such a covering $M = \bigcup_{i=1}^r E_i$ of $M$ by submodules and let $M$ be cyclic. For a generator  $x \in M$ there exists $i \in \myset{r}$ such that $x \in E_i$ and thus $M = Rx \subseteq E_i \subset M$ that leads to a contradiction.
	
	If $M$ is not cyclic, then for any $x \in M\setminus\{0\}$, $\{0\} \subset Rx \subset M$ and therefore $M = \bigcup_{x \in M\setminus\{0\}} Rx$.
\end{proof}

\begin{lemma}\label{lemma-nontrivial-iff-noncyclic}
	For each non-cyclic module $M$ there exists a nontrivial solution of \cref{eq-main-space-equation} with at least one module equals $M$.
	A solution of the equation
\begin{equation*}
\sum_{i = 1}^s a_i \id_{V_i} = \sum_{i = 1}^t b_i \id_{U_i}\;,
\end{equation*} 
	with only cyclic modules is trivial, where all the coefficients are in $\mathbb{C}$.
\end{lemma}
\begin{proof}
	Prove the first part.
	Let $M = \bigcup_{i=1}^r E_i$ be a nontrivial covering of $M$ by submodules. Denote $M_I = \bigcap_{i \in I} E_i$, where $I \subseteq \myset{r}$ and define $M_{\emptyset} = M$. Use the inclusion-exclusion formula,
	\begin{equation*}
	\sum_{\card{I} \text{ is even}} \id_{M_I} = \sum_{\card{I} \text{ is odd}} \id_{M_I}\;,
	\end{equation*}
	where the summation is over all subsets $I \subseteq \myset{r}$. 
	It is easy to see that the resulting equation is nontrivial, for example, the module $M$ appears only from the left side. The number of terms on each side is the same and equals $2^{r - 1}$.
	
Prove the second part. Assume that there exists a nontrivial solution of the equation.
Without loss of generality, we can assume that the equation is simplified by eliminating equal terms and making a reindexing. Hence, $a_i,b_j \neq 0$ for $i \in \myset{s}$, $j \in \myset{t}$ and all $V_i, U_i$ are different.
Since the solution is nontrivial, $s,t>0$. Among the modules choose the maximal with respect to the inclusion, suppose it is $V_1$. Then $V_1 = \bigcup_{i=1}^t (V_1 \cap U_i)$, where $\{0\} \subset V_1 \cap U_i \subset V_1$, $i \in \myset{t}$. From \Cref{lemma-cyclic-if-covering}, the module $V_1$ is therefore non-cyclic, which contradicts to our assumption.
\end{proof}

\paragraph{Characters and Fourier transform.} 
Denote by $\hat{A} = \Hom_\mathbb{Z}(A, \mathbb{C}^*)$ the set of characters of $A$. The set $\hat{A}$ has a natural structure of a right $R$-module.
Let $A, W$ be two left $R$-modules. For a map $\sigma \in \Hom_R({W},{A})$ define a map $\hat{\sigma}: \hat{A} \rightarrow \hat{W}$, $\chi \mapsto \chi \sigma$. Note that $\hat{\sigma} \in \Hom_R(\hat{A}_R,\hat{W}_R)$. It is known that $\wedge$ is an exact contravariant functor on the category of left(right) $R$-modules, see \cite{wood-foundations}.

Let $M$ be a left $R$-module. The Fourier transform of a map $f: M \rightarrow \mathbb{C}$ is a map $\mathcal{F}(f): \hat{M} \rightarrow \mathbb{C}$, defined as
\begin{equation*}
\mathcal{F}(f)(\chi) = \sum_{m \in M} f(m)\chi(m)\;.
\end{equation*}
It can be easily proved that for a submodule $V \subseteq M$, $\mathcal{F}(\id_V) = \card{V} \id_{V^\perp}$, where an orthogonal module is defined as $V^\perp = \{ \chi \in \hat{M} \mid \forall v \in V, \chi(v) = 1 \} \subseteq \hat{M}$.
Note that the Fourier transform is invertible, $V^{\perp\perp} \cong V$ and for any $V, U \subseteq W$, $(V \cap U)^\perp = V^\perp + U^\perp$. 

For any $\sigma \in \Hom_R(W,A)$, $\Ker \sigma = \{w \in W \mid \sigma(w) = 0 \} = \{ w \in W \mid \forall \chi \in \hat{A}, \chi(\sigma(w)) = 1 \} = (\Img \hat{\sigma})^\perp$, and thus $(\Ker \sigma)^\perp = \Img \hat{\sigma}$.
\begin{theorem}\label{thm-cyclic-extendable}
	If $\hat{A}$ is a cyclic right $R$-module and $A$ is pseudo-injective, then $A$ has the extension property.
\end{theorem}
\begin{proof}
	Let $C \subset A^n$ be a code and let $f \in \Hom_R(C,A^n)$ be an isometry. By \Cref{thm-isometry-criterium}, $f$ is extendable if and only if the solution $(\U,\V)$ of \cref{eq-main-space-equation} is trivial.
	Due to the properties of the Fourier transform, \cref{eq-main-space-equation} is equivalent to the following equality of functions defined on $\hat{W}$,
		\begin{equation}\label{eq-dual}
		\sum_{i = 1}^n \card{V_i} \id_{V_i^\perp} = \sum_{i = 1}^n \card{U_i} \id_{U_i^\perp}\;
		\end{equation}
	and the solution of \cref{eq-main-space-equation} is trivial if and only if the corresponding orthogonal solution is trivial.
	The statement of the theorem is a direct consequence of \Cref{lemma-nontrivial-iff-noncyclic} and the fact that the modules $V_i^\perp = \Img \hat{\lambda_i}$, $U_i^\perp = \Img \hat{\mu_i}$, $i \in \myset{n}$, are all cyclic, since so is $\hat{A}$.
\end{proof}

\begin{remark}
	\Cref{thm-cyclic-extendable} is an analogue of \Cref{thm-wood-general-modules}, where instead of the cyclic socle condition we use the cyclic character module condition. Prove that these two conditions are equivalent. In \cite{wood-foundations} it was proven that $\soc(A)$ is cyclic if and only if $A$ can be embedded into $_R \hat{R}$. This means there exists an injective homomorphism of left $R$-modules $\phi : A \rightarrow _R \hat{R}$. Since $\wedge$ is an exact functor, the last is equivalent to the fact that the map $\hat{\phi}: \hat{\hat{R}}_R \cong R_R \rightarrow \hat{A}_R$ is a projective homomorphism of right $R$-modules that is a characterization of cyclicity of $\hat{A}_R$. \comment{Note that $\Ker \phi = A^\perp$ and hence $R_R / A^\perp \cong \hat{A}_R$.}
\end{remark}

\section{Extension theorem for matrix alphabets}\label{sec-main}
An $R$-module $A$ is called \emph{semisimple} (or completely reducible) if $A$ is a direct sum of simple submodules.
\begin{lemma}\label{lemma-semisimple-pseudoinjective}
	If a left $R$-module $A$ is semisimple, then $A$ is pseudo-injective. 
\end{lemma}
\begin{proof}
	Let $N, M \subseteq A$ be two submodules and let $\psi: N \rightarrow M$ be an isomorphism. Since $A$ is semisimple, there exist $N', M' \subseteq A$ such that $A = N \oplus N' = M \oplus M'$. Since $N \cong M$, there is an isomorphism $\phi: N' \rightarrow M'$. Then $\psi$ extends to the automorphism $\psi\times \phi: A = N \oplus N' \rightarrow M \oplus M' = A$.
\end{proof}

It is proved in \cite[p. 656]{lang} that each module $M$ over the ring $R = \M_m(\mathbb{F}_q)$ is semisimple and is isomorphic to $\M_{m \times k}(\mathbb{F}_q)$ for some $k$. Call $k$ the dimension of $M$ and denote $\dim M = k$.
We need the following lemmas to prove an extension theorem for $R$-linear codes over $M$.

\begin{lemma}\label{lemma-binomial-sums}
	The following equalities hold,
	\begin{equation*}
	\sum_{i = 0}^{t-1} (-1)^i q^{\binom{i}{2}} \binom{t}{i}_q = (-1)^{t-1} q^{\binom{t}{2}}\;,
	\end{equation*}
	\begin{equation*}
	\sum_{i = 0}^t q^{\binom{i}{2}} \binom{t}{i}_q = \prod_{i=0}^{t-1} (1 + q^i)\;.
	\end{equation*}
\end{lemma}
\begin{proof}
	Use a well-known Cauchy binomial theorem,
	\begin{equation*}\prod_{i=0}^{t-1} (1 + x q^i) = \sum_{i = 0}^t q^{\binom{i}{2}} \binom{t}{i}_q x^i\;.
	\end{equation*}
	To get the equalities in the statement, put $x = -1$ and $x = 1$.
\end{proof}

\begin{lemma}\label{lemma-number-of-submodules}
	Let $R$ be a matrix module over $\mathbb{F}_q$, let $M$ be a $t$-dimensional $R$-module and let $X$ be a $p$-dimensional submodule of $M$. For each  $i \in \{p,\dots, t\}$ we have,
	\begin{equation*}
	\card{ \{ V \subseteq M \mid X \subseteq V, \dim V = i \} } = \binom{t - p}{i - p}_q
	\end{equation*}
\end{lemma}
\begin{proof}
	It is a well known fact that for any $k \in \{0, \dots, t\}$ there are $\binom{t}{k}_q$ submodules of $M$ of dimension $i$. The number of such submodules, that contain $X$ is equal to the number of submodules of dimension $i - \dim X$ in $M/X$ that contain $\{0\}$, which is always the case. In other words, $\card{ \{ V \subseteq M \mid X \subseteq V, \dim V = i \} }  = \card{ \{ V \subseteq M/X \mid \dim V = i - p \} } = \binom{\dim M/X}{i - p}_q = \binom{t - p}{i - p}_q$.
\end{proof}

The next proposition is an improvement of \Cref{thm-wood-matrix-module}. Note that the code length $K$ in \Cref{thm-wood-matrix-module} depends on $k$ (the alphabet parameter) whereas in the following proposition the code length $N$ depends on $m$ (the ring parameter) and $N$ is not greater than $K$. In our proof this improvement is easily obtained from the author's construction in \cite{wood-foundations}, however it does not appear in the original statement of the author. 
\begin{proposition}\label{thm-matrix-module-bound}
	Let $R = \M_{m}(\mathbb{F}_q)$ and let $A = \M_{m \times k}(\mathbb{F}_q)$ be a left $R$-module.
	
	If $k \leq m$, then the alphabet $A$ has the extension property.
	
	If $k > m$, there exist a linear code $C \subset A^N$, $N = \prod_{i=1}^{m} (1 + q^i)$, and a map $f \in \Hom_R(C,A^N)$ that is a Hamming isometry, but there is no monomial transformation extending $f$.
	
	For any $n < N$, each Hamming isometry $f \in \Hom_R(C,A^n)$ is extendable.
\end{proposition}
\begin{proof}
For any $k$, since $A$ is semisimple, by \Cref{lemma-semisimple-pseudoinjective}, $A$ is pseudo-injective.
If $k\leq m$, the right $R$-module $\hat{A}$ is cyclic, since $\dim \hat{A} = \dim A = k \leq m$. From \Cref{thm-cyclic-extendable}, $A$ has the extension property.

To construct a code of the length $N$ we do the following. Let $C'$ be the code over the alphabet $B = M_{m \times m+1}(\mathbb{F}_q)$ and let $f \in \Hom_R(C',B^N)$ be the unextendable isometry from \Cref{thm-wood-matrix-module}. Choosing this alphabet, we have $K = N$. 
Since $k > m$, in $A$ there exists a submodule isomorphic to $B$, so $C'$ can be considered as a code in $A^N$ and $f \in \Hom_R(C', A^N)$. Due to the construction of the author in \cite{wood-foundations}, the code $C'$ has all zero column and $f(C')$ does not. Therefore $f$ is unextendable.

Let $k > m$.
Let $n'$ be the minimum value of the code length for which there exists and unextendable isometry $f \in \Hom_R(C,A^{n'})$. By \Cref{thm-isometry-criterium}, there exists a non-trivial solution of \cref{eq-main-space-equation}.
Hence, the minimum length $n$ of a nontrivial solution of \cref{eq-main-space-equation} is not greater than $n'$.
Consider a solution of the length $n$ which have the minimum value $\max\{ \dim V_i, \dim U_i \mid i \in \myset{n} \}$, and denote this value by $r$.
From \Cref{lemma-nontrivial-iff-noncyclic}, $r >m$.
Without loss of generality, let $\dim {V_1} = r$.

Introduce a new notation. Denote $I_j = \{ i \mid \dim V_i < r - j \}$, $J_j = \{ i \mid \dim U_i < r - j \}$ and
\begin{equation*}
\Sigma_j= \sum_{\dim V = r - j} \id_V\;,
\end{equation*}
for $j \in \{ 0, \dots, r \}$, where the summation is over all the submodules in $V_1$ of the given dimension.
Calculate the restriction of \cref{eq-main-space-equation} on the module $V_1$,
\begin{equation*}
a\Sigma_0 = \sum_{i \in J_0} \id_{U_i \cap V_1} - \sum_{i \in I_0} \id_{V_i \cap V_1}\;,
\end{equation*}
where by $a \geq 1$ we denote the number of modules $V_1$ in the left part of \cref{eq-main-space-equation}.
This is a nontrivial solution of the length $n$ and the maximum dimension $r$. Evidently, since the length $n$ is the minimal, $U_i \cap V_1 \subset V_1$ for all $i \in J_0$ and $V_i \cap V_1 \subset V_1$ for all $i \in I_0$, so, without loss of generality, let $U_i, V_i \subseteq V_1$, for $i \in \myset{n}$.

Say that on the $t$-step, $0 \leq t < r$, we have proved that \cref{eq-main-space-equation} is of the form,
\begin{equation*}
	a \sum_{i = 0}^t (-1)^iq^{\binom{i}{2}} \Sigma_i  =  \sum_{i \in J_t} \id_{U_i} - \sum_{i \in I_t} \id_{V_i}\;.
\end{equation*}
For $t = 0$ this is true.
Let $X \subset V_1$ be of dimension $r - t - 1$. Restrict the equation on $X$,
\begin{align*}
a \sum_{i = 0}^t (-1)^i q^{\binom{i}{2}} \sum_{\dim V = r - i}\id_{V \cap X} 
= \sum_{i \in J_{t}} \id_{U_i\cap X} - \sum_{i \in I_{t}} \id_{V_i\cap X} \;.
\end{align*}
The dimension of $X$ is smaller than $r$, so the restricted solution is trivial.
Calculate the number of $\id_X$ terms from the left and from the right.
Denote  $b = \card{\{ i \in J_t \mid X = U_i \}}$ and $c = \card{\{ i \in I_t \mid X = V_i \}}$. Note that either $b=0$ or $c=0$. Using \Cref{lemma-number-of-submodules} and \Cref{lemma-binomial-sums},
\begin{equation*}
a \sum_{i = 0}^t (-1)^i q^{\binom{i}{2}} \binom{t+1}{i}_q = a(-1)^t q^{\binom{t+1}{2}} \binom{t+1}{i}_q = b - c\;,
\end{equation*}
and therefore $c = 0$ if $t$ is even and $b = 0$ if $t$ is odd.

All the submodules of $V_1$ of dimension $r - t - 1$ are presented from the left or from the right side of \cref{eq-main-space-equation}, depending on the parity of $t$, with the same multiplicity. Considering this fact, we rewrite \cref{eq-main-space-equation} in the form,
\begin{equation*}
a \sum_{i = 0}^{t+1} (-1)^i q^{\binom{i}{2}} \Sigma_i = \sum_{i \in J_{t+1}} \id_{U_i} - \sum_{i \in I_{t+1}} \id_{V_i}\;.
\end{equation*}
On the output of $t = r - 1$ step we get,
\begin{equation*}
a \sum_{i = 0}^{r} (-1)^i q^{\binom{i}{2}} \Sigma_i = \sum_{i \in J_{r}} \id_{U_i} - \sum_{i \in I_{r}} \id_{V_i} \equiv 0\;.
\end{equation*}
The length of the equation is $\frac{1}{2}a\sum_{i = 0}^{r} q^{\binom{i}{2}} \binom{r}{i}_q$. Since the equation has the minimal length, $a = 1$ and $r = m + 1$.
From \Cref{lemma-binomial-sums}, $n = \frac{1}{2} \prod_{i=0}^{m} (1 + q^i) = \prod_{i=1}^{m} (1 + q^i) = N$. Therefore $n' \geq n = N$.
\end{proof}

\section{Extension theorem for MDS codes}
There is a famous Singleton bound, that states that for a code $C \subseteq A^n$, $\card{C} \leq \card{A}^{n - d + 1}$, where $d$ is the minimum distance of $C$. When a code $C$ attains the bound, it is called an MDS code. The value $k = n - d +1$ is called a dimension $C$ and the code is said to be an $(n,k)_A$ MDS code.

An alternative definition is the following. A code $C \subseteq A^n$ is MDS if and only if a restriction of $C$ on any $k$ columns is isomorphic to $A^k$. In other words, any $k$ columns of $C$ can be taken as an information set of the code. We interpret this definition in terms of modules $\V = (V_1,\dots,V_n)$.
\begin{lemma}\label{lemma:mds-modules}
	Let $C$ be an $(n, k)_A$ MDS code. For each subset $I \subseteq \myset{n}$ of size $k$, $\sum_{i \in I} V_i^{\perp} = \hat{W}_R$. Moreover, $\card{V_i^\perp} = \card{A}$ for all $i \in \myset{n}$.
\end{lemma}
\begin{proof}
	Let $I \subseteq \myset{n}$ be a subset with $k$ elements. Let $C'$ be a code obtained from $C$ by keeping only coordinates from $I$. The map $\lambda' = (\lambda_i)_{i \in I}$, $\lambda': W \rightarrow A^k$ is a parametrization of $C'$. Since $C$ is MDS, $\lambda'$ is injective, which implies $\bigcap_{i \in I} V_i = \{0\}$. Calculating the orthogonal, we get $\sum_{i \in I} V_i^\perp = \hat{W}_R$.
	
	We know that all the modules $W,C,C'$ are isomorphic to $A^k$. Thus there is an isomorphism of right $R$-modules $\hat{W}\cong \hat{A}^k$. Also, $\card{V_i^\perp} \leq \card{A} = \card{\hat{A}_R}$ and for any $i,j \in \myset{n}$, $\card{V_i^\perp + V_j^\perp} = \card{V_i^\perp}\card{V_j^\perp}/\card{V_i^\perp \cap V_j^\perp}$. Combining all the facts, we get $\card{V_i^\perp} = \card{A}$.
\end{proof}

The next lemma shows that the condition of pseudo-injectivity in \Cref{thm-isometry-criterium} can be omitted if a code is MDS.

\begin{lemma}\label{lemma:mds-pseudoinjectivity}
		Let $C$ be an $(n, k)_A$ MDS code and let $f \in \Hom_R(C,A^n)$. If $\V = \U$, then $f$ extends to a monomial map.
\end{lemma}
\begin{proof}
	The proof is almost identical to the second part of the proof of \Cref{thm-isometry-criterium}.
	Let $\sigma, \tau \in \Hom_R(W,A)$ be two maps that parametrize a column in $C$ and a column in $f(C)$ correspondingly. Since $C$ is an MDS code, from \Cref{lemma:mds-modules}, $\Img \sigma = A$, because $\card{\Img \sigma} = \card{(\Ker \sigma)^\perp} = \card{A}$.
	
	Let $\Ker \sigma = \Ker \tau = N \subseteq W$. This implies $\Img \tau = \Img \sigma = A$.
	Consider the canonical isomorphisms $\bar{\sigma}, \bar{\tau}: W/N \rightarrow A$.
	The map $h\in \Aut_R(A)$, defined as $h = \bar{\tau}\bar{\sigma}^{-1}$, satisfies the equality $h \sigma = \tau$.
\end{proof}

\begin{theorem}\label{thm-mds-extension-theorem}
	Let $R$ be a ring with identity and let $A$ be a finite left $R$-module. Let $C$ be an $(n, k)_A$ MDS code, $k\neq 2$. Each Hamming isometry $f \in \Hom_R(C,A^n)$ extends to a monomial map. 
\end{theorem}
\begin{proof}
	Assume that there exists an unextendable isometry $f \in \Hom_R(C,A^n)$. From \Cref{thm-isometry-criterium} and \Cref{lemma:mds-pseudoinjectivity}, there exists a nontrivial solution of \cref{eq-main-space-equation}, or equivalently, there exists a nontrivial solution of the orthogonal equation (\ref{eq-dual}). It is clear that $f(C)$ is also an MDS code.
	
	The proof is obvious for the case $k = 1$, so let $k \geq 3$. This means, from \Cref{lemma:mds-modules}, for any different $i,j,k \in \myset{n}$, $V_i^\perp \cap (V_j^\perp + V_k^\perp) = \{0\}$.
	Without loss of generality, assume that $U_1^\perp$ is covered nontrivially by modules $V_1^\perp,\dots, V_t^\perp$, $t>1$, i.e. $U_1^\perp = \bigcup_{i=1}^t V_i^\perp$, $\{0\} \subset V_i^\perp \subset U_1^\perp$, for $i \in \myset{t}$ and no module is contained in another.

	Take a nonzero element $a \in U_1^\perp \cap V_1^\perp$ and a nonzero element $b \in U_1^\perp \cap V_2^\perp$. Obviously, since $V_1^\perp \cap V_2^\perp = \{0\}$, $a + b \not\in V_1^\perp \cup V_2^\perp$.
	But $a+b \in U_1^\perp$ and hence $t>2$. There exists an index $i$, let it be $3$, such that $a + b \in U_1^\perp \cap V_3^\perp$. Then $a+b \in (V_1^\perp + V_2^\perp) \cap V_3^\perp \neq \{0\}$, which gives a contradiction.
\end{proof}

The case of MDS codes of dimension 2 is observed in \cite{d3}, where $R$ is a finite field and the alphabet $A$ is a vector space. Note that the statement is true for all abelian groups as $\mathbb{Z}$-modules. In \cite{forney} the author proved that there exists only $(n,1)_G$ and $(n,n)_G$ MDS codes over a nonabelian group $G$. It is not difficult to show that an analogue of the extension property holds for these two families of trivial codes.

\end{document}